\pdfoutput=1
\documentclass[10pt,pdftex,a4paper]{article}%
\usepackage{float}
\floatstyle{boxed} 
\restylefloat{figure}
\usepackage{subfig}

\usepackage{verbatim}
\usepackage{amssymb}
\usepackage{lipsum}
\usepackage{multicol}
\usepackage{amsmath}
\usepackage{framed}
\usepackage{amsthm}
\usepackage{tikz}
\usepackage[margin=0.92in]{geometry}
\usepackage{color}
\usepackage{cite}
\usepackage{tabularx}
\newif\ifcode
\codefalse
\codetrue
\usepackage{macros}
\ifcode
\usepackage[ruled]{algorithm}
\usepackage{algpseudocode}
\usepackage{ioa_code}

\newtheorem{theorem}{Theorem}

\newtheorem{definition}{Definition}
\newtheorem{lemma}[theorem]{Lemma}

\newcounter{linenumber}
\newcommand{\RNum}[1]{\uppercase\expandafter{\romannumeral #1\relax}}

\def\P{\ensuremath{\mathcal{P}}}

\def\T{\ensuremath{\mathcal{T}}}

\def\S{\ensuremath{\mathcal{S}}}
\def\D{\ensuremath{\mathcal{D}}}

\def\Nat{\ensuremath{\mathbb{N}}}

\def\dd{search}

\newcommand{\LS}{LS}

\newcommand{\true}{\lit{true}}
\newcommand{\false}{\lit{false}}

\newcommand{\remove}[1]{}

\newcommand{\id}[1]{\mbox{\textit{#1}}}% for identifiers in code
\newcommand{\LL}{\ms{LL}}
\newcommand{\ignore}[1]{}

\begin{document}
\bibliographystyle{abbrv}

\title{In the Search of Optimal Concurrency}

\author{
Vincent Gramoli$^3$~~~Petr Kuznetsov$^1$~~~Srivatsan Ravi$^2$ \\
$^1$\normalsize T\'el\'ecom ParisTech \\
$^2$\normalsize Purdue University \\
$^3$\normalsize University of Sydney
}

%%%%%%%%%%%%%%%%%%%%%%%%%%%%%%%%%%%%%%%%%%%%%%%%%%%%%%%%%%%%%%%%%%%%%%%%%%%%%%%%
\date{}
\maketitle

\renewcommand*\thesection{\arabic{section}}
\renewcommand*\thesubsection{\arabic{section}.\arabic{subsection}}

\begin{abstract}
Implementing a concurrent data structure typically begins with
defining its sequential specification.
However, when used \emph{as is}, a nontrivial sequential data structure, such as 
a linked list, a search tree, or a hash table, may expose incorrect behavior: lost
updates, inconsistent responses, etc.
To ensure correctness, portions of the sequential code operating on
the shared data must be ``protected'' from data races using
synchronization primitives and, thus, certain schedules
of the steps of concurrent operations must be rejected.
But can we ensure that we do not ``overuse'' synchronization, \emph{i.e.}, that we
reject a concurrent schedule only if it violates correctness?   

In this paper, we treat this question formally by 
introducing the
notion of a \emph{concurrency-optimal} implementation.
A program's concurrency is defined here as its ability to accept 
concurrent schedules, \emph{i.e.}, interleavings of 
steps of its sequential implementation.
An implementation is concurrency-optimal if it accepts all 
interleavings that do not violate the program's correctness.
We explore the concurrency properties of 
\emph{search} data structures which can be represented in the form of directed
acyclic graphs exporting insert, delete and search operations. 
We prove, for the first time, that \emph{pessimistic} (\emph{e.g.}, based on
conservative locking) and \emph{optimistic serializable} (\emph{e.g.}, based
on serializable transactional memory) implementations of search data-structures
are incomparable in terms of concurrency.
Specifically, there exist simple interleavings of sequential code
that cannot be accepted by \emph{any} pessimistic (and \emph{resp.}, serializable
optimistic) implementation,
but accepted by a serializable optimistic one (and \emph{resp.}, pessimistic).
Thus, neither of these two implementation classes is concurrency-optimal. 
\end{abstract}

\ignore{
\begin{abstract}
Modern concurrent programming benefits from a large variety of synchronization techniques.
These include conventional pessimistic locking, as well as optimistic
techniques based on conditional synchronization instructions or 
transactional memory. Yet, it is unclear which of these approaches 
better leverage the concurrency inherent to multi-cores.

In this paper, 
we compare formally ``the amount of concurrency''
one can obtain by converting a  sequential program into a concurrent
one using optimistic or pessimistic synchronization techniques. 
%To establish fair comparison of such implementations, 
We introduce a new correctness
criterion for concurrent programs, defined independently 
of the synchronization techniques they use. 
We treat a program's concurrency as its ability to accept a
concurrent schedule, a metric inspired by the theories of both databases
and transactional memory. 
%[[ VG: bogus sentence
%We consider two synchronization techniques: conservative (pessimistic)
%locking and optimistic serializable
%.
We consider two classes of synchronization techniques: 
pessimistic ones that conservatively protect data items before accessing them and optimistic ones that 
may have to roll back and re-access the items.

We explore the concurrency properties of a large class of 
data structures which can be represented in the form of directed
acyclic graphs, and that we call \emph{\dd{}} structures.
We show that, for some workloads, serializable optimistic {\dd} implementations,
%[[VG: the workloads is defined by the Dag structures
%, for some workloads, 
%]]
allow for more concurrency than \emph{any} pessimistic implementation.
We also show that there exist pessimistic {\dd} implementations that, for other workloads, allow for more 
concurrency than \emph{any} serializable optimistic implementation.
We derive that a \emph{concurrency-optimal} implementation 
must combine the advantages of pessimistic implementations, 
namely their semantics awareness,  and the advantages 
of optimistic implementations, namely their ability to 
restart operations in case of conflicts.
\end{abstract}
}

%
%\thispagestyle{empty}
%\clearpage
\pagenumbering{arabic}

%%%%%%%%%%%%%%%%%%%%%%%%%%%%%%%%%%%%
%%%%%%%%%%%%%%%%%%%%%%%%%%%%%%%%%%%%
%%% Local Variables:
%%% mode: latex
%%% mode: flyspell
%%% Local IspellDict: "american"
%%% mode: outline-minor
%%% End:
%!TEX root = 0main.tex
\section{Introduction}
\label{sec:intro}
Building concurrent abstractions that efficiently exploit 
multi-processing abilities of modern hardware is a challenging task.
One of the most important issues here is the need for \emph{synchronization}, \emph{i.e.},
ensuring that concurrent processes resolve conflicts on the shared
data in a consistent way.
Indeed, using, \emph{as is}, a data structure designed for conventional sequential
settings in a concurrent environment may cause different kinds of
inconsistencies, and synchronization must be used to avoid them.  
A variety of synchronization techniques have been developed so far.
Conventional pessimistic synchronization 
conservatively protects shared data with locks before reading or
modifying it.
Speculative synchronization, achieved using transactional memory (TM) or
conditional instructions, such as CAS or LL/SC,  
optimistically executes memory operations with a risk of aborting them in the future.
A programmer typically uses these synchronization techniques 
to ``wrap'' fragments of a sequential implementation of the desired
abstraction (\emph{e.g.}, \emph{operations} on the shared data structure)
in a way that allows the resulting concurrent execution to ``appear''
\emph{locally sequential}, 
while preserving a \emph{global} correctness criterion.

It is however difficult to tell in advance 
which of the 
techniques will provide more \emph{concurrency}, \emph{i.e.}, which one would allow
the
resulting programs to process more 
%[[VG: not more operations but more executions
executions of concurrent operations 
without data conflicts. 
Implementations based on TMs~\cite{ST95,norec}, which execute concurrent accesses speculatively, 
may seem more concurrent than lock-based counterparts whose concurrent accesses are blocking. 
But TMs conventionally impose
\emph{serializability}~\cite{Pap79-serial} or even
stronger properties~\cite{tm-book} on operations 
encapsulated within transactions. 
%[[VG: not clear what "they" refers to
%making sure that they constitute a
%correct sequential execution.
%to make sure that each operation execution constitutes a
%correct sequential execution.
%]] 
This may prohibit certain concurrent
scenarios allowed by a large class of dynamic data
structures~\cite{GG14}. 

In this paper, we reason formally about the  ``amount of
concurrency'' one can obtain 
by turning a sequential program into a concurrent one. 
To enable fair comparison of different synchronization techniques,    
we (1)~define what it means for a concurrent program to be
correct regardless of the type of  synchronization it uses and 
(2)~define a metric of concurrency. 
These definitions allow
us to compare concurrency properties offered by serializable optimistic and pessimistic
synchronization techniques,  whose popular examples are, respectively, transactions
and conservative locking.

\paragraph{Correctness}
Our  novel consistency criterion, called \emph{locally-serializable
linearizability}, is an intersection of \emph{linearizability} and 
a new \emph{local serializability} criterion.

Suppose that we want to design a concurrent implementation of a data
type $\T$ (\emph{e.g.}, \textsf{integer set}), given its sequential
implementation $\S$
(\emph{e.g.}, based on a sorted linked list). 
A concurrent implementation of  $\T$ is
\emph{locally serializable} with respect to $\S$ if it
ensures that the local execution of \emph{reads} and \emph{writes} of each 
operation is, in precise sense,  equivalent to \emph{some} execution of $\S$.
This condition is weaker than serializability
since it does not require the existence of  a \emph{single} sequential 
execution that is consistent with all local executions.
It is however sufficient to guarantee that
executions do not observe an inconsistent transient state that could 
lead to fatal arithmetic errors, \emph{e.g.}, division-by-zero.

In addition, for the implementation of $\T$ to ``make sense'' globally, 
%the high-level history of 
every concurrent execution should be  
\emph{linearizable}~\cite{HW90,AW04}: the invocation and responses of
high-level operations observed in the execution
should constitute a correct sequential history
of $\T$.
The combination of local serializability and linearizability gives
a correctness criterion that we call \emph{\LS-linearizability},
where {\LS} stands for ``locally serializable''.
We show that LS-linearizability,
just like linearizability, is
%is , as the original  linearizability,
\emph{compositional}~\cite{HW90,HS08-book}: a composition of LS-linearizable 
implementations is also LS-linearizable. 
%[[PK
% Unlike linearizability, however, LS-linearizability is \emph{not}
% non-blocking~\cite{HW90,HS08-book}: 
% local serializability may prevent an operation in a finite LS-linearizable
% history from completing in a non-blocking manner. 

%[[PK not the right place
\ignore{
We apply the criterion of LS-linearizability to  
two broad classes of synchronization techniques: 
\emph{pessimistic} and \emph{serializable optimistic}. 
Pessimistic implementations capture what can be achieved 
using classic conservative locks like mutexes, 
spinlocks, reader-writer locks.
% \petr{why ``wait-free'' here? do we have examples of
%   wait-free LS-linearizable algos?} 
Optimistic implementations, however proceed speculatively and
may roll back in the case of conflicts.
Additionally, \emph{serializable} optimistic techniques, \emph{e.g.},  
relying on conventional TMs, like TinySTM~\cite{FFR08} or NOrec~\cite{norec}
allow for transforming any sequential implementation of a data type to a LS-linearizable concurrent one.
Note that, pessimistic implementations may sometimes sacrifice
serializability by fine-tuning to the semantics of the data type.
}

% , or more relaxed forms of
% optimistic techniques, such as ``lazy''
% synchronization~\cite{HHL+05} or elastic transactions~\cite{FGG09}.
%[[ SR: included wait-free implementations
%\petr{Is this sufficient to fight the critique on ``equating optimism
%  with naive TMs''?}
\paragraph{Concurrency metric}
%\paragraph{Concurrency metric.}
We measure the amount of concurrency provided by an LS-linearizable implementation as the set of schedules it accepts.
To this end, we define a concurrency metric 
inspired by the analysis of parallelism in database concurrency control~\cite{Yan84,Her90}
and transactional memory~\cite{GHF10}.
More specifically, we assume an external scheduler that defines which
processes execute which steps of the corresponding sequential program 
in a dynamic and unpredictable fashion. 
This allows us to define concurrency provided by an implementation as the set of \emph{schedules} 
(interleavings of reads and writes of concurrent sequential operations) 
it \emph{accepts} (is able to effectively process).
%
%[[PK tautology?
%We use this metric to derive  the amounts of concurrency provided by different 
%classes of concurrent implementations of a given sequential program
%from the sets of schedules they accept.
%]]

Our concurrency metric is platform-independent and it allows for
measuring relative concurrency of LS-linearizable implementations
using arbitrary synchronization techniques.   
%[[PK should go to the discussion
\ignore{
We do not claim that this metric necessarily captures 
efficiency, as it does not account for other factors, 
like cache sizes, cache coherence protocols, or computational costs of 
validating a schedule, which may also affect performance on
multi-core architectures.
}
%However, our preliminary performance evaluations show that more
%concurrency typically means more throughput.\petr{Is it fair to say?}
%[[PK redundant
%While our metric does not aim at capturing these parameters, 
%]]
% In a concurrent submission (also available as a technical
% report~\cite{CONCURP-TR}), we show that the search can result in 
% a \emph{concurrency-optimal} implementation
% exhibiting unprecedented performance:
% our \emph{concurrency-optimal} implementation of a
% listed-based set outperforms all analogous algorithms known so far.   

The combination of our correctness and concurrency definitions
provides a framework to compare the concurrency one can get
by choosing a particular synchronization technique for a specific data type.

\paragraph{Measuring concurrency: pessimism vs. serializable optimism}
%[[PK redundant
\ignore{
For the first time, we analytically capture the
inherent incomparability of serializable optimistic vs. pessimistic
implementations when it comes to exploiting concurrency.
We show that serializable optimistic implementations of a large class of 
data structure workloads
allow for more concurrency than \emph{any} pessimistic implementation.
We also show that there exist pessimistic implementations that, for other workloads, allow for more 
concurrency than \emph{any} serializable optimistic implementation.
}
We explore the concurrency properties of a large class of
\emph{search} concurrent data structures.
%The considered class of data structures, called \emph{{\dd}
%structures}, implement a dictionary abstraction 
%as they
Search data structures maintain data 
in the form of a rooted directed
acyclic graph (DAG), where each node is a $\tup{\textit{key},\textit{value}}$
pair, and export operations \emph{insert}$(\textit{key},\textit{value})$,
\emph{delete}$(\textit{key})$, and \emph{find}$(\textit{key})$ 
%\emph{read} and \emph{write} 
with the natural sequential semantics. 
%PK: trivial?
\ignore{
An \emph{insert}$(\textit{key})$ checks whether the
%[[VG: inverting for simplifying
dictionary already contains a node equipped with the given key, and 
if so, returns $\false$, otherwise it creates a new node, links it to the graph and returns $\true$.
A \emph{remove}$(\textit{key})$ operation checks whether the
dictionary contains a node with the given key, unlinks it
from the graph and returns $\true$, otherwise returns $\false$. 
A \emph{find}$(\textit{key})$ operation returns the pointer of a node
with the given key if such a node is found and returns $\false$
otherwise.
}

A typical sequential implementation of these operations
traverses the graph starting from the root in order to locate the
``relevant'' area corresponding to its $\textit{key}$ parameter.
If the operation is \textit{insert} or \textit{delete} and the relevant area satisfies certain
conditions, the operation updates it.
The class includes many popular data structures, such as linked lists, skiplists, and
search trees, implementing various abstractions like sets and
multi-sets.
Under reasonable assumptions on the representation of data, search
structures are ``concurrency-friendly'': 
high-level operations may commute (and, thus, proceed concurrently and independently) 
if their relevant sets are disjoint.
To exploit this potential concurrency, various synchronization techniques can be
employed.

In this paper, we compare the concurrency properties of two classes of
search-structure implementations: \emph{pessimistic} and \emph{serializable optimistic}.
Pessimistic implementations capture what can be achieved 
using classic conservative locks like mutexes, 
spinlocks, reader-writer locks.
In contrast, optimistic implementations, however proceed speculatively and
may roll back in the case of conflicts.
Additionally, \emph{serializable} optimistic techniques, \emph{e.g.},  
relying on conventional TMs, like TinySTM~\cite{FFR08} or NOrec~\cite{norec}
allow for transforming any sequential implementation of a data type to a LS-linearizable concurrent one.
%Note that, pessimistic implementations may sometimes sacrifice
%serializability by fine-tuning to the semantics of the data type.

The main result of this paper is that synchronization techniques based on pessimism and 
serializable optimism, are not concurrency-optimal:
we show that no one of their respective set of accepted concurrent schedules include the other. 

On the one hand, we prove, for the first time, that there exist simple schedules
that are not accepted by \emph{any} pessimistic implementation,
but accepted by a serializable optimistic implementation.
Our proof technique, which is interesting in its own right, is based on
the following intuitions: a pessimistic implementation 
has to proceed irrevocably and over-conservatively reject 
a potentially acceptable schedule, simply because it \emph{may} result
in a data conflict thus leading the data structure to an inconsistent
state. However, an optimistic implementation of a
search data structure may (partially or completely) restart an operation 
depending on the current schedule.  
This way even schedules that potentially lead to conflicts may be
optimistically accepted. 

On the other hand, we show that pessimistic implementations can be designed to exploit
the semantics of the data type.
In particular, they can allow operations updating disjoint sets of
data items to proceed independently and preserving linearizability of
the resulting history, even though the execution is not serializable.    
In such scenarios, pessimistic implementations carefully adjusted to
the data types we implement can supersede the 
``semantic-oblivious'' optimistic serializable implementations.
Thus, neither pessimistic nor serializable optimistic implementations
are concurrency-optimal.

Our comparative analysis of concurrency properties of pessimistic and
serializable optimistic implementation suggests 
that combining the advantages of pessimism, 
namely its semantics awareness, and the advantages 
of optimism, namely its ability to 
restart operations in case of conflicts, 
enables implementations that are strictly 
better-suited for exploiting concurrency 
than any of these two techniques taken individually.
% For example, \cite{GKR15} describes a \emph{concurrency-optimal} optimistic implementation 
% of \emph{list-based set} that also empirically outperforms
% state-of-the-art algorithms~\cite{HHL+05,harris-set,michael-set}.     
To the best of our knowledge, this is the first formal analysis of
the relative abilities of different synchronization techniques to exploit concurrency
in dynamic data structures and lays the foundation for designing concurrent data structures
that are concurrency-optimal.

%% parallel paper~\cite{CONCURP-TR}.
%% SR: this does not seem to convey any information
\paragraph{Roadmap}
We define the class of concurrent implementations we consider in Section~\ref{sec:prel}.
In Section~\ref{sec:concurrency}, we define the correctness criterion and our concurrency metric.
Section~\ref{sec:dirds} defines the class of \dd{} structures for which our concurrency lower bounds apply.
In Section~\ref{sec:pessimistic}, we analyse the concurrency provided by pessimistic and serializabile optimistic synchronization
techniques in exploiting concurrency w.r.t \dd{} structures.
Sections~\ref{sec:related} and \ref{sec:conc} present the related work and concluding remarks respectively.
%
%
%%%%%%%%%%%%%%%%%%%%%%%%%%%%%%%%%%%%%%
%!TEX root = 0main.tex
\section{Preliminaries}
\label{sec:prel}
\paragraph{Sequential types and implementations}
%\paragraph{Sequential types and implementations.}
%\subsection{Sequential types and implementations}
%
An \emph{type} $\tau$ is a tuple
$(\Phi,\Gamma, Q, q_0, \delta)$ where
$\Phi$ is a set of operations,
$\Gamma$ is a set of responses, $Q$ is a set of states, $q_0\in Q$ is an
initial state and 
$\delta \subseteq Q\times \Phi \times Q\times \Gamma$ 
is a \emph{sequential specification} that determines, for each state
and each operation, the set of possible
resulting states and produced responses~\cite{AFHHT07}. 
%[[ V: the following sentence introduces multiple terms, refer to \tau as multiple types...
%We consider types $\tau$ for which 
%there exists a \emph{sequential
%implementation} $\id{IS}$ 
%that implements a \emph{high-level object} $O_{\tau}$. 

Any type $\tau=(\Phi,\Gamma, Q, q_0, \delta)$ is associated with a \emph{sequential implementation}
$\id{IS}$.
The implementation encodes states in $Q$ using a collection of elements
$X_1,X_2,\ldots $ and,  for each operation of $\tau$, specifies a
sequential \emph{read-write} algorithm.
Therefore, in the implementation $\id{IS}$, an operation performs a
sequence of  \emph{reads} and \emph{writes} on 
$X_1,X_2,\ldots$ and returns a response $r\in \Gamma$. 
The implementation guarantees that, when executed sequentially, starting from the state of
$X_1,X_2,\ldots$ encoding $q_0$, the operations eventually return
responses satisfying $\delta$.
% For example, we depict the \emph{list-based set} in Appendix~\ref{app:seq}: the \emph{set} type 
% implemented as a \emph{sorted linked list}.

\paragraph{Concurrent implementations}
We consider an asynchronous shared-memory
system in which a set of processes communicate by
applying \emph{primitives} on shared \emph{base objects}~\cite{Her91}.
%[[PK16 do not see the point here
%We place no upper bounds on the number of versions an object may maintain or on the size of this object.
%]]

We tackle the problem of turning the sequential
algorithm $\id{IS}$ of type $\tau$ into a \emph{concurrent} one, shared by 
$n$ processes $p_1,\ldots,p_n$ ($n\in\Nat$).
The idea is that the concurrent algorithm essentially follows
$\id{IS}$, but to ensure correct operation under concurrency, it
replaces read and write operations on 
$X_1,X_2,\ldots$ in operations of $\id{IS}$ with their base-object
implementations. 

%We refer to the resulting implementation as a concurrent implementation of $(\id{IS},\tau)$.
Throughout this paper, we use the term \emph{operation} to refer to
high-level operations of the type.
Reads and writes implemented by a concurrent algorithm are referred simply 
as \emph{reads} and \emph{writes}.
Operations on base objects are referred to as \emph{primitives}.

We also consider concurrent implementation that 
execute portions of sequential code  \emph{speculatively}, and restart
their operations when conflicts are encountered.  
To account for such implementations, we assume that an implemented
read or write may \emph{abort} by returning a special response
$\bot$. In this case, we say that the corresponding (high-level)
operation is \emph{aborted}. 
%[[PK16 how does it matter?
%The $\bot$ event is treated both as the response event of the read or
%write operation and as the response of the corresponding high-level operation.   
%]]

Therefore, our model applies to all concurrent algorithms 
in which a high-level operation can be seen as a
sequence of reads and writes on elements $X_1,X_2,\ldots$ (representing
the state of the data structure), with the option of aborting the
current operation and restarting it after.
Many existing concurrent data structure implementations comply with this model
as we illustrate below. 

\paragraph{Executions and histories}
An \emph{execution} of a concurrent implementation is a sequence
of invocations and responses of high-level operations of type $\tau$, %(atomic) 
invocations and responses of read and write
operations, and invocations and responses of base-object primitives.
%Throughout this paper,  the term \emph{operation} refers to some
%high-level operation of the type, 
%while reads and writes on objects are referred as \emph{read} (and resp. \emph{write}) operations.
%%\emph{synchronization} primitives (e.g., lock acquisitions and
%%releases or transaction delimiters). 
%%A \emph{step} in an execution corresponds to one of the above events.
We assume that executions are \emph{well-formed}:
no process invokes a new read or write, or high-level operation before
the previous read or write, or a high-level operation, resp., 
returns, or takes steps outside its 
%[[PK16 is this important? 
%read or write 
%]]
operation's interval.

Let $\alpha|p_i$ denote the subsequence of an execution $\alpha$
restricted to the events of process $p_i$.
Executions $\alpha$ and $\alpha'$ are \emph{equivalent} if for every process
$p_i$, $\alpha|p_i=\alpha'|p_i$.
An operation $\pi$ \emph{precedes} another operation $\pi'$ in an execution
$\alpha$, 
denoted $\pi \rightarrow_{\alpha} \pi'$, 
if the response of $\pi$ occurs before the invocation of $\pi'$.
Two operations are \emph{concurrent} if neither precedes
the other. 
%An execution is \emph{rw-sequential} if every invocation of a read or write operation is immediately followed by a response event.
An execution is \emph{sequential} if it has no concurrent 
operations. 
A sequential execution $\alpha$ is \emph{legal} 
if for every object $X$, every read of $X$ in $\alpha$ 
returns the latest written value of $X$.
An operation is \emph{complete} in $\alpha$ if the invocation event is
followed by a \emph{matching} (non-$\bot$) response or aborted; otherwise, it is \emph{incomplete} in $\alpha$.
Execution $\alpha$ is \emph{complete} if every operation is complete in $\alpha$.

The \emph{history exported by an execution $\alpha$} is
the subsequence of $\alpha$ reduced to the invocations and responses
of operations, reads and writes, except for the reads
and writes that return $\bot$ (the abort response). 

\paragraph{High-level histories and linearizability}
%\paragraph{High-level histories and linearizability.}
%
A \emph{high-level history} $\tilde H$ of an execution $\alpha$ is the subsequence of $\alpha$ consisting of all
invocations and responses of \emph{non-aborted} operations.
A complete high-level history $\tilde H$ is \emph{linearizable} with 
respect to an object type $\tau$ if there exists
a sequential high-level history $S$ equivalent to $H$ such that
(1) $\rightarrow_{\tilde H}\subseteq \rightarrow_S$ and
(2) $S$ 
is consistent with the sequential specification of type $\tau$.
Now a high-level history $\tilde H$ is linearizable if it can be
\emph{completed} (by adding matching responses to a subset of
incomplete operations in $\tilde H$ and removing the rest)
to a linearizable high-level history~\cite{HW90,AW04}.

\paragraph{Optimistic and pessimistic implementations}
%\paragraph{Pessimistic implementations.}
Note that in our model an implementations may, under
certain conditions, abort an operation:
some read or write return $\bot$,
in which case the corresponding operation also returns $\bot$.
Popular classes of such \emph{optimistic} implementations are those based on
``lazy synchronization''~\cite{HHL+05,HS08-book} (with the ability of
returning $\bot$ and re-invoking an operation) or   
transactional memory (\emph{TM})~\cite{ST95,norec}.

In the subclass of \emph{pessimistic} implementations, 
%the exported history contains
%every read-write event that appears in the execution. 
%More precisely, 
no execution includes operations that return $\bot$.  
Pessimistic implementations are typically \emph{lock-based} or based
on pessimistic TMs~\cite{PLE12}.
%\textbf{SR: do we need to describe locks? can we just give a reference to mutexes?}
A lock provides 
%shared or 
exclusive (resp., shared) access to an element $X$ through synchronization primitives
$\lit{lock}(X)$ (resp., $\lit{lock-shared}(X)$), 
%(\emph{shared mode}),
%$\lit{lock}(X)$ (\emph{exclusive mode}),  
and $\lit{unlock}(X)$ (resp., $\lit{unlock-shared}(X)$).
A process \emph{releases} the lock it holds by invoking
$\lit{unlock}(X)$ or $\lit{unlock-shared}(X)$.  
When $\lit{lock}(X)$ 
%(resp. $\lit{lock}(X)$) 
invoked by a process $p_i$ returns, we say that $p_i$ \emph{holds
a lock on $X$} (until $p_i$ returns from the subsequent $\lit{lock}(X)$).  
When $\lit{lock-shared}(X)$ 
%(resp. $\lit{lock}(X)$) 
invoked by $p_i$ returns, we say that $p_i$ \emph{holds
a shared lock on $X$} (until $p_i$ returns from the subsequent $\lit{lock-shared}(X)$). 
%in shared (resp. exclusive) mode}.
At any moment, at most one process may hold a lock on an element $X$.
Note that two processes can hold a shared lock on $X$ at a time.
We assume that locks are \emph{starvation-free}: if no process holds a
lock on $X$ forever, then every $\lit{lock}(X)$
eventually returns. 
Given a sequential implementation of a data type, 
a corresponding lock-based concurrent one 
is derived by inserting the synchronization primitives ($\lit{lock}$
and $\lit{unlock}$) to protect read and write accesses to the shared data. 
%]]
%\vspace{1mm}\noindent\textbf{Optimistic implementations.}
%\paragraph{Optimistic implementations.}

%%%%%%%%%%%%%%%%%%%%%%%%%%%%%%%%%%%%%%%
%!TEX root = 0main.tex
\section{Correctness and concurrency metric}
\label{sec:concurrency}
In this section, we define the correctness criterion of \emph{locally serializable linearizability}
and introduce the framework for comparing the relative abilities of different synchronization technique in 
exploiting concurrency.
\subsection{Locally serializable linearizability}
%[[PK16
%We are now ready to define the correctness criterion that we impose on our
%concurrent implementations.
%]]
Let $H$ be a history and let $\pi$ be a high-level operation in $H$. 
Then $H|\pi$ denotes the subsequence of $H$ consisting of the events
of $\pi$, except for the last aborted read or write, if any.
Let $\id{IS}$ be a sequential implementation of an object of type
$\tau$ and $\Sigma_{\id{IS}}$, the set of histories of $\id{IS}$. 
\begin{definition}[LS-linearizability]
\label{def:lin}
A history ${H}$ is \emph{locally serializable with respect to}
${\id{IS}}$ if for every high-level operation $\pi$ in $H$,
there exists $S \in \Sigma_{\id{IS}}$ such that $H|\pi=S|\pi$.

A history ${H}$ is \emph{\LS-linearizable with respect to
$(\id{IS},\tau)$}  (we also write $H$ is $(\id{IS},\tau)$-LSL)  if:
(1) ${H}$ is locally serializable with respect to
$\id{IS}$ and (2) the corresponding high-level history $\tilde H$ 
is linearizable with respect to $\tau$.
%\begin{enumerate}
%\item ${H}$ is locally serializable with respect to
%$\id{IS}$ and
%\item the corresponding high-level history $\tilde H$ 
%is \emph{linearizable with respect to $\tau$}.
%\end{enumerate}
\end{definition}
Observe that local serializability stipulates that the execution is 
witnessed sequential by every operation.
%in isolation [[PK why in isolations???]] 
Two different operations (even when invoked by the same process) are not
required to witness mutually consistent sequential executions.

A concurrent implementation $I$ is \emph{\LS-linearizable with respect to
$(\id{IS},\tau)$} (we also write $I$ is $(\id{IS},\tau)$-LSL)
if every history exported by $I$ is $(\id{IS},\tau)$-LSL.   
Throughout this paper, when we refer to a concurrent implementation of $(\id{IS},\tau)$, 
we assume that it is \LS-linearizable with respect to $(\id{IS},\tau)$.

\paragraph{Compositionality}
We define the composition of two distinct object types $\tau_1$ and $\tau_2$ 
as a type $\tau_1\times\tau_2=(\Phi,\Gamma,Q,q_0,\delta)$ as follows: 
$\Phi=\Phi_1\cup \Phi_2$, $\Gamma=\Gamma_1\cup 
\Gamma_2$,\footnote{Here we treat each $\tau_i$ as a distinct type by adding
index $i$ to all elements of $\Phi_i$, $\Gamma_i$, and $Q_i$.}   
$Q=Q_1\times Q_2$,
$q_0=({q_0}_1,{q_0}_2)$, and  $\delta \subseteq Q\times \Phi \times Q\times
\Gamma$ is such that $((q_1,q_2),\pi,(q_1'q_2'),r)\in\delta$ if and only if for $i\in \{1,2\}$, if 
$\pi\in \Phi_i$ then $(q_i,\pi,q_i',r)\in\delta_i$ $\wedge$ $q_{3-i}=q^{\prime}_{3-i}$.

Every sequential implementation $\id{IS}$ of an object  $O_1\times O_2$ of a
composed type $\tau_1\times\tau_2$ naturally induces two sequential
implementations $I_{S1}$ and $I_{S2}$ of objects $O_1$ and $O_2$,
respectively. 
Now a correctness criterion 
$\Psi$
is \emph{compositional} if for every
history $H$ on an object composition $O_1\times O_2$, 
if 
$\Psi$
holds for $H|O_i$ with
respect to $I_{Si}$, for $i \in \{1,2\}$, then
$\Psi$
holds for $H$ with
respect to $\id{IS}=I_{S1}\times I_{S2}$.
Here, $H|O_i$ denotes the subsequence of $H$ consisting of events on $O_i$.
\begin{theorem}
\label{th:comp:app}
\LS-linearizability is compositional. 
\end{theorem}
\begin{proof}
Let $H$, a history on $O_1\times O_2$,  be \LS-linearizable
with respect to $\id{IS}$. 
Let each $H|O_i$, 
$i\in\{1,2\}$, 
be \LS-linearizable with respect to $I_{Si}$. 
Without loss of generality,  we assume that $H$ is complete (if $H$
is incomplete, we consider any completion of it containing
\LS-linearizable completions of  $H|O_1$ and $H|O_1$).

Let $\tilde H$ be a completion of the high-level history corresponding to $H$ such that
%\vincent{Not, the completed history of $H$?}
$\tilde H|O_1$ and $\tilde H|O_2$ are linearizable with respect to $\tau_1$
and $\tau_2$, respectively. Since linearizability is
compositional~\cite{HW90,HS08-book}, $\tilde H$ is linearizable with respect to $\tau_1\times\tau_2$.

Now let, for each operation $\pi$, $S_{\pi}^1$ and $S_{\pi}^2$ be any two sequential histories of
$I_{S1}$ and $I_{S2}$  such that $H|\pi|O_j=S_{\pi}^j|\pi$, for $j \in \{1,2\}$
(since 
$H|O_1$
and $H|O_2$ are \LS-linearizable such histories exist).
We construct a sequential history $S_{\pi}$ by interleaving events of
$S_{\pi}^1$ and $S_{\pi}^2$ so that $S_{\pi}|O_j=S_{\pi}^j$, 
$j\in\{1,2\}$.
Since each $S_{\pi}^j$ acts on a distinct component $O_j$ of $O_1\times
O_2$, every such $S_{\pi}$ is a sequential history of $\id{IS}$.
We pick one $S_{\pi}$ that respects the local history $H|\pi$,
which is possible, since $H|\pi$ is consistent with both    
$S_1|\pi$ and $S_2|\pi$. 

Thus, for each $\pi$, we obtain a history of $\id{IS}$ that agrees with
$H|\pi$. Moreover, the high-level history of $H$ is linearizable. Thus, $H$ is \LS-linearizable with respect to $\id{IS}$. 
\end{proof}
%
%[[PK16 worth noticing
Note that LS-linearizability is not non-blocking~\cite{HW90,HS08-book}: local
serializability may prevent an operation in a finite LS-linearizable
history from having a completion, e.g., because, it might read an
inconsistent system state caused by a concurrent incomplete operation.  
%]]

\paragraph{LS-linearizability versus other criteria}
%[[PK put in a more general context]]
LS-linearizability is a two-level consistency criterion which makes it
suitable to compare concurrent implementations of a sequential data
structure, regardless of synchronization techniques they use.
It is quite distinct from related criteria designed for database and software
transactions, such as serializability~\cite{Pap79-serial,WV02-book} and
multilevel serializability~\cite{Wei86,WV02-book}.

For example, serializability~\cite{Pap79-serial} prevents sequences of reads and writes from conflicting in a cyclic way, 
establishing a  global order of transactions.
Reasoning only at the level of reads and writes may be overly conservative:
higher-level operations may commute even if their reads and writes conflict~\cite{Wei88}.
Consider an execution of a concurrent \emph{list-based set} depicted in 
Figure~\ref{fig:ex1}.
% [[ V: {1,2,3} is a set state not elements
We assume here that the set initial state is $\{1,3,4\}$.
%linked list initially contains elements $\{1,3,4\}$.
%]]
Operation $\lit{contains}(5)$ is concurrent, first with 
operation $\lit{insert}(2)$ and then with operation $\textsf{insert}(5)$. 
The history is not serializable:
%, i.e.,
%there is no global serial history compatible with the local executions:
$\lit{insert}(5)$ sees 
the effect of $\lit{insert}(2)$ because $R(X_1)$ by $\lit{insert}(5)$ returns the value
of $X_1$ that is updated by $\lit{insert}(2)$ and
thus should be serialized after it. But $\lit{contains}(5)$ misses
element $2$ in the linked list, but must read the value of $X_4$ that is updated by
$\lit{insert}(5)$ to perform the read of $X_5$, \emph{i.e.}, the element created by $\lit{insert}(5)$.  
However, this history is LSL since each of the three local histories is consistent with some
sequential history of $\LL$. 

Multilevel serializability~\cite{Wei86,WV02-book} was 
proposed to reason in terms of multiple semantic levels in the same execution.
\LS-linearizability, being defined for two levels only, does not require a global serialization of low-level operations as
$2$-level serializability does. 
LS-linearizability simply requires each process  to observe a local serialization, which can be different from one
process to another. Also, to make it more suitable for concurrency
analysis of a concrete data structure, instead of semantic-based commutativity~\cite{Wei88}, we use the sequential
specification of the high-level behavior of the object~\cite{HW90}.

Linearizability~\cite{HW90,AW04} only accounts for high-level
behavior of a data structure,  so it does not imply
LS-linearizability. For example, Herlihy's universal
construction~\cite{Her91} provides a linearizable implementation for
any given object type, but does not guarantee that each execution locally appears
sequential with respect to any sequential implementation of the type.    
Local serializability, by itself, does not require any synchronization
between processes and can be trivially implemented without
communication among the processes.
%, so that each process is only aware
%of its own operations.   
Therefore, the two parts of LS-linearizability indeed complement each other.  
% Note also that we can easily think of an implementation that is linearizable but
% not LS-linearizable (does not look sequential locally) and another
% that is locally serializable but not linearizable 
% (does not make sense
% globally). An example of the former is Herlihy's universal
% construction~\cite{Her91}, and
% [[[\textbf{SR: we should expand why?}]]]
% an example of the latter is a ``bogus'' implementation that assumes no
% communication among the threads. 
% Therefore, the two parts of LS-linearizability indeed complement each other.
%
\subsection{Concurrency metric}
To characterize the ability of a concurrent implementation to process arbitrary interleavings of sequential code, we introduce 
the notion of a \emph{schedule}.
Intuitively, a schedule describes the order in which complete high-level
operations, and sequential reads and writes are invoked by the user. 
More precisely, a schedule is 
an equivalence class of complete histories that agree on
the \emph{order} of invocation and response events of reads, writes and high-level operations, but 
not necessarily on the responses of read operations or 
of high-level operations.
%possibly not on read values or high-level responses.
%]]
Thus, a schedule can be treated as a history, where responses of read
and high-level operations
are not specified. 
%[[PK16

We say that an implementation $I$ \emph{accepts} a schedule $\sigma$ if 
it exports a history $H$ such that $\ms{complete}(H)$ exhibits
the order of $\sigma$, where $\ms{complete}(H)$ is the subsequence of $H$
that consists of the events of the complete operations that returned a matching response. 
We then say that the execution (or history) \emph{exports} $\sigma$. 
A schedule $\sigma$ is 
%LS-linearizable with respect to 
$(\ms{IS},\tau)$-LSL if there
exists an $(\id{IS},\tau)$-LSL history exporting $\sigma$.
% and that is \LS-linearizable w.r.t $(\id{IS},\tau)$.
%%\LS-linearizable w.r.t $(\id{IS},\tau)$ history that exports $\sigma$.

An $(\id{IS},\tau)$-LSL implementation is therefore
\emph{concurrency-optimal} if it accepts all  $(\ms{IS},\tau)$-LSL schedules.

%[[PK moved to Sec 5
\ignore{

%[[PK cut space
%We now characterize the relative power of \emph{synchronization techniques}
%in transforming a sequential implementation of a data type into a concurrent one.
%]]
A \emph{synchronization technique} is a set of concurrent implementations.
% [[PK: too long
%We now define two synchronization techniques: one encompasses a subclass of
%optimistic implementations while the other characterizes pessimistic implementations.
We define below a specific optimistic synchronization technique and then
a specific pessimistic one.
% ]]

\paragraph{The class $\mathcal{SM}$}
%\paragraph{The class $\mathcal{SM}$.}
Let $\alpha$ denote the execution of a concurrent implementation and
$\ms{ops}(\alpha)$, 
the set of operations each of which performs at least one event in $\alpha$.
%Let $\ms{cseq}(\alpha)$ be the subsequence of a finite execution $\alpha$ derived by removing all 
%events from every transaction that is either incomplete or aborted in $\alpha$.
Let ${\alpha}^k$ denote the prefix of $\alpha$ up to the last event of operation $\pi_k$.
Let $\ms{Cseq}(\alpha)$ denote the set of subsequences of ${\alpha}$  that
consist of all the events of operations that are complete in $\alpha$. 
We say that $\alpha$ is \emph{strictly serializable} if 
there exists a legal sequential execution $\alpha'$ equivalent to
a sequence in $\sigma\in\ms{Cseq}(\alpha)$
such that $\rightarrow_{\sigma} \subseteq \rightarrow_{\alpha'}$. 

This paper focuses on optimistic implementations that are strictly
serializable and, in addition, guarantee that every operation (even aborted or incomplete) observes correct (serial)
behavior.  
More precisely, an execution $\alpha$ is  \emph{safe-strict serializable} if
(1) $\alpha$ is strictly serializable, and
(2) for each operation $\pi_k$, % that is incomplete or returned
%$\bot$ in $\alpha$, 
there exists a legal sequential execution
$\alpha'=\pi_0\cdots \pi_i\cdot \pi_k$ and
$\sigma\in\ms{Cseq}(\alpha^k)$ such that $\{\pi_0,\cdots, \pi_i\}
\subseteq \ms{ops}(\sigma)$ and $\forall \pi_m\in \ms{ops}(\alpha'):{\alpha'}|m={\alpha^k}|m$.

Safe-strict serializability captures nicely both local serializability
and linearizability. 
If we transform a sequential implementation
$\id{IS}$ of a type $\tau$ into a \emph{safe-strict serializable} concurrent one, 
we obtain an LSL implementation of $(\id{IS},\tau)$. 
% Thus, the following lemma is immediate.
% %
% \begin{lemma}
% Let $I$ be a safe-strict serializable implementation of $(\id{IS},\tau)$.
% Then, $I$ is \LS-linearizable with respect to $(\id{IS},\tau)$.
% \end{lemma}
%
Indeed, 
% by running each operation of $\id{IS}$ within a transaction of
% a safe-strict serializable TM, 
we make sure that completed operations witness the same execution of $\id{IS}$, and
every operation that returned $\bot$ is consistent with some
execution of $\id{IS}$ based on previously completed operations. 
Formally, $\mathcal{SM}$ denotes the set of optimistic, safe-strict serializable LSL
implementations.

\paragraph{The class $\mathcal{P}$}
This denotes the set of \emph{deadlock-free} pessimistic
LSL implementations: assuming that no process stops taking steps of
its algorithm in the middle of a high-level operation,
at least one of the concurrent operations return a matching response~\cite{HS11-progress}.
Note that $\mathcal{P}$ includes implementations that are not necessarily safe-strict serializable. 
}
\section{Search data structures}\label{sec:dirds}
%\section{Dag-dictionary data structures}\label{sec:dirds}
%
In this section, we introduce a class $\D$ of \emph{dictionary-search}
data structures (or simply \emph{search structures}),
inspired by the study of \emph{dynamic databases} undertaken by
Chaudhri and Hadzilacos~\cite{CH98}. 
%We discuss 

%\paragraph{Data representation.}
\paragraph{Data representation}
At a high level, a {\dd} structure is a dictionary that maintains data in a directed acyclic graph (DAG) with a
designated \emph{root} node (or element). The
vertices (or nodes) of the graph are key-value pairs and edges specify the
\emph{traversal function}, \emph{i.e}, paths that should be taken by the dictionary
operations in order to find the nodes that are going to determine the
effect of the operation.
%(we call such nodes \emph{relevant}). 
Keys are natural numbers, values are taken from a set $V$ and the outgoing edges of each node are locally labelled. 
By a light abuse of notation we say that $G$ \emph{contains} both nodes and edges.
Key values of nodes in a DAG $G$ are related by a partial order $\prec_{G}$ 
that additionally defines a property $\mathbb{P}_G$ specifying if there is an outgoing edge from node with key $k$ to a node 
with key $k'$ (we say that $G$ \emph{respects} $\mathbb{P}_G$).

%[[PK does not seem to apply to BST, need to revise
If $G$ contains a node $a$ with key $k$, 
the \emph{$k$-relevant} set of $G$, denoted $V_k(G)$,  
is $a$ plus all nodes $b$, such that $G$ contains $(b,a)$ or $(a,b)$.
%(By a light abuse of notation we say that $G$ \emph{contains} both nodes and edges.)  
If $G$ contains no nodes with key $k$, 
$V_k(G)$ consists of all nodes $a$ of $G$ with the smallest $k\prec_{G} k'$ plus all
nodes $b$, such that $(b,a)$ is in $G$.
The \emph{$k$-relevant graph of $G$}, denoted $R_k(G)$, is the
subgraph of $G$ that consists of all paths from the root to the nodes
in $V_k(G)$.

%The \emph{$k$-relevant path} is the shortest length path from the root to
%a node in the $k$-relevant set. Intuitively, any operation with
%parameter $k$ should traverse at least the $k$-relevant or a longer path. 

%We say that an edge in a graph of key-value nodes is \emph{valid} 
%goes from a node with smaller key to a node with a larger key. 
%The graph is \emph{valid} if its every edge is valid.
%As a result, each valid directed graph of key-value nodes is acyclic. 
%We assume that structures in $\D$ maintain valid DAGs only. 

%Let $G$ be a valid DAG and $G'$ be a \emph{subDAG} of
%$G$, \emph{i.e.}, a DAG whose nodes and edges are subsets of nodes and
%edges of $G$.
%The \emph{boundary} of $G'$ (with respect to $G$) is the set of nodes
%of $G'$ that have no edges to other nodes in $G'$.

\begin{figure}
  \includegraphics[scale=0.45]{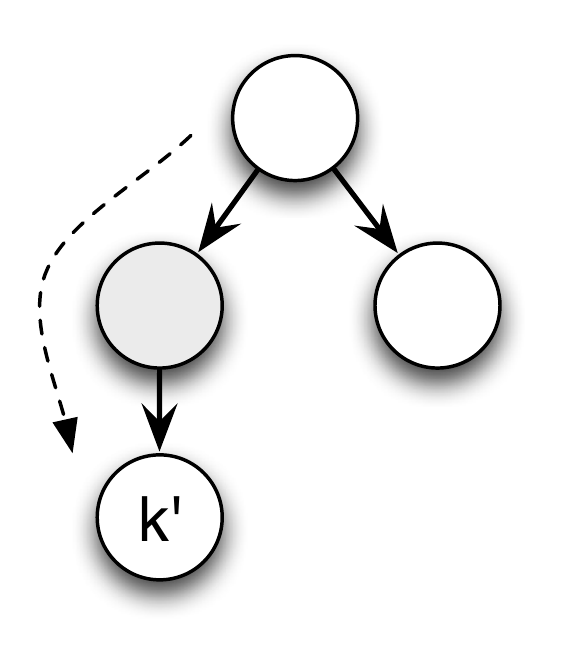}
  \includegraphics[scale=0.45]{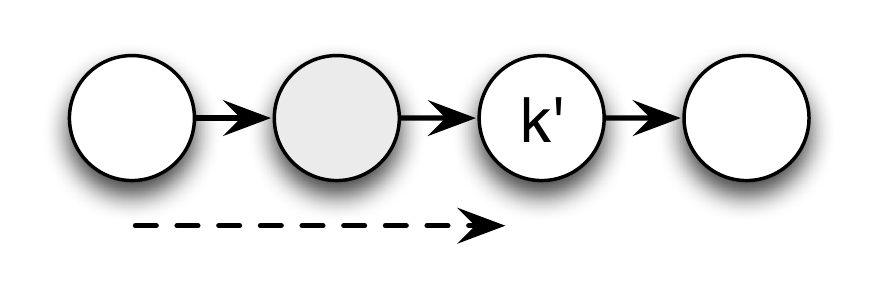}
  \includegraphics[scale=0.45]{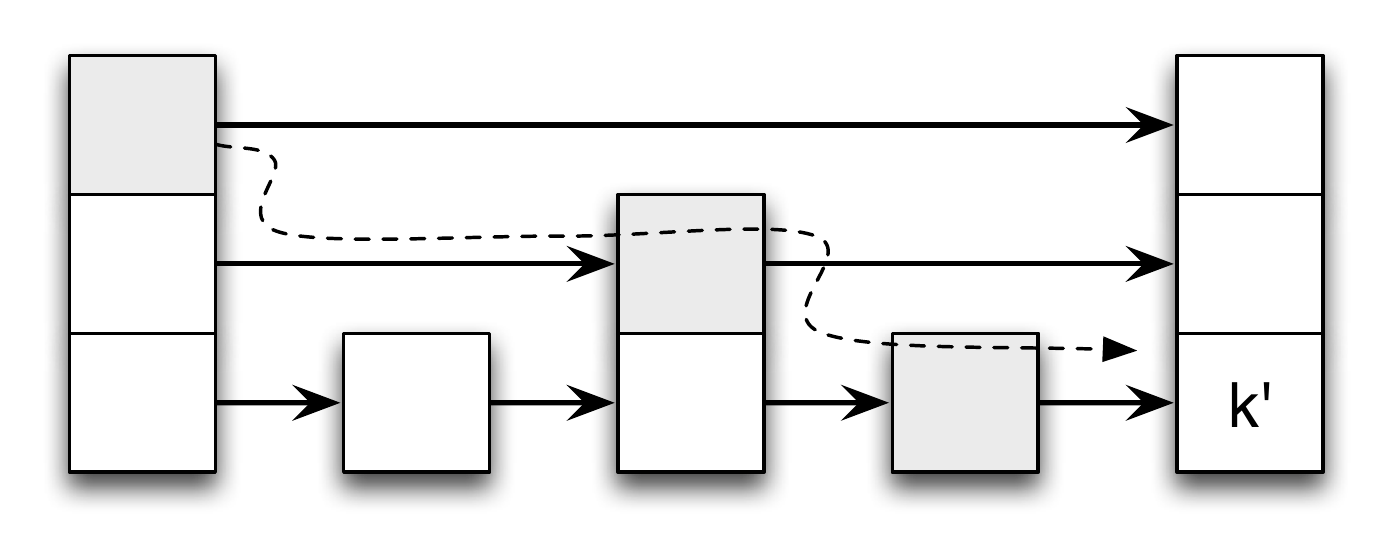}
  \caption{Three \dd{} structures, a binary tree, a linked list and a skip list, whose k-relevant set has grey nodes and whose k-relevant path is indicated with a dashed arrow\label{fig:dag-ds}}
\end{figure}
\paragraph{Sequential specification}
At a high level, every data structure in $\D$ exports a sequential
specification with the following operations:
\begin{itemize}
\item  \textit{insert}$(k,v)$ checks whether a node with key $k$ is already
  present and, if so, returns $\false$, otherwise it creates a node with key $k$ and value $v$,
  %and some initial value, 
  \emph{links} it to the graph (making it reachable from the
  root) and returns $\true$;    
\item \textit{delete}$(k)$ checks whether a node with key $k$ is already
  present and, if so, \emph{unlinks} the node from the graph (making
  it unreachable from the root) 
 and returns $\true$, otherwise it returns $\false$;    
\item \textit{find}$(k)$ returns the pointer to the node with key $k$ 
or $\false$ if no such node is found.   
%\item \textit{read}$(k)$ returns the value of the node with key $k$
% or $\false$ if no such node is found;
%\item \textit{write}$(v,k)$ replaces the value of the node with $k$
 % and returns $ok$ or $\false$ if no such node is found.   
%
\end{itemize} 
\paragraph{Traversals}
For each operation
$op\in\{\ms{insert}(k,v),\ms{delete}(k),$ $\ms{find}(k)\}_{k\in
  \Nat, v\in V}$, each search structure is parameterized by  
a (possibly randomized) \emph{traverse function} $\tau_{op}$.
Given the \emph{last visited} node $a$ and the DAG of already visited nodes $G_{op}$,
the traverse function $\tau_{op}$ returns a new node  
$b$ to be \emph{visited}, \emph{i.e.}, accessed to get its $key$ and the list of descendants, 
or $\emptyset$ to indicate that the search is complete.

\paragraph{Find, insert and delete operations}
Intuitively, the traverse function is used by the operation $op$ to explore
the {\dd} structure and, when the function returns $\emptyset$, the sub-DAG $G_{op}$
explored so far contains enough information for
operation $op$ to complete. 

If $op=\textit{find}(k)$,
$G_{op}$ either contains a node with key $k$ or ensures that the
whole graph does not contain $k$. 
As we discuss below, in \emph{sorted} search structures, such as sorted
linked-lists  or skiplists, we can stop as soon as all outgoing edges in $G_{op}$
belong to nodes with keys $k'\geq k$.  Indeed, the remaining nodes can
only contain keys greater than $k$, so $G_{op}$ contains enough
information for $op$ to complete.   

%\textbf{PK: Use the traverse function to define $k$-relevant sets...}
An operation $op=\textit{insert}(k,v)$,   
is characterized by an \emph{insert function} $\mu_{(k,v)}$ that, given
a DAG $G$ and a new node $\tup{k,v}\notin G$,  returns the set
of edges from nodes of $G$ to $\tup{k,v}$ and from $\tup{k,v}$ to nodes
of $G$ so that the resulting graph is a DAG containing
$\tup{k,v}$ and respects $\mathbb{P}_{G}$.

An operation $op= \textit{delete}(k)$,   
is characterized by a \emph{delete function} $\nu_{k}$ that, given
a DAG $G$, gives the set of edges to be removed and a set of
edges to be added in $G$ so that the resulting graph is a DAG that respects $\mathbb{P}_{G}$.
 
%\paragraph{Sequential implementations.}
\paragraph{Sequential implementations}
We make the following natural  assumptions on the sequential implementation of a
{\dd} structure:

\begin{itemize}
\item \textit{Traverse-update} Every operation $op$ starts with the read-only
  \emph{traverse} phase followed with a write-only \emph{update} phase.  
The traverse phase of an operation $op$ with parameter $k$ completes at the latest when
for the visited nodes $G_{op}$ contains the $k$-relevant graph.
The update phase of a \emph{find}$(k)$ operation is empty. 

\item
\textit{Proper traversals and updates} For all DAGs $G_{op}$ and nodes $a\in G_{op}$, the traverse function
$\tau_{op}(a,G_{op})$ returns $b$ such that $(a,b)\in G$.
The update phase of an \emph{insert}$(k)$ or \emph{delete}$(k)$ operation modifies outgoing edges 
of $k$-relevant nodes.

\item
\textit{Non-triviality} There exist a key $k$ and a state $G$ such that 
(1) $G$ contains no node with key $k$, 
(2) If $G'$ is the state resulting after applying $\lit{insert}(k,v)$
to $G$, then there is exactly one edge $(a,b)$ in $G'$ such that
$b$ has key $k$, and (3)~the
shortest path in $G'$ from the root to $a$ is of length at least $2$.

%\item[] 
%\item[Local consistency.] For each DAG $G$, there exists a
%  sequential execution that results in $G$.
%The traversal function 
% $\tau_{\textit{find}(k)}$ is defined on each valid DAG.    
%
\end{itemize}
The non-triviality property says that in some cases the
read-phase may detect the presence of a given key only at the last
step of a traverse-phase. Moreover, it excludes the pathological DAGs
in which all the nodes are always reachable in one hop from the root.
%
%The local consistency
%property implies that any execution of a traverse phase of an
%operation $op$ 
%%traverses the DAG by applying its $\tau_{op}$ operation 
%is consistent with some sequential execution, even in the presence of
%concurrent updates.     
%
Moreover, the traverse-update property and the fact that keys are
natural numbers implies that every traverse
phase eventually terminates.  Indeed, there can be only finitely many
vertices pointing to a node with a given key, thus, eventually a
traverse operation explores enough nodes to be sure that no node with
a given key can be found.   
%
%\vincent{Shouldn't we mention this after presenting the skip list and the tree in the next paragraph?}
%\textbf{Something about the benefits of skiplists and trees.}     
%SR: I agree
 %\vspace{2mm}\noindent\textbf{Commutativity and hand-over-hand locking.}
%
%The following properties are going to be instrumental in our
%subsequent concurrency analysis.   Let $A$ be any abstraction in $\D$.

%\begin{observation}
%Applying a state of $A$ in which a node $k$ is reachable, two
%\emph{insert}$(v)$ operations do not commute.
%\end{observation}

%\paragraph{Examples.}
\paragraph{Examples of {\dd} data structures}
In Figure~\ref{fig:dag-ds}, we describe few data structures in $\D$.
A \emph{sorted linked list} maintains a single path,
starting at the \textit{root} sentinel node and ending at a \textit{tail}
sentinel node, and any traversal with parameter $k$ simply follows the path until a node
with key $k'\geq k$ is located. 
The traverse function for all operations follows the only path
possible in the graph until the two relevant nodes are located. 
%For details, we refer to Appendix~\ref{app:seq}.    

A \emph{skiplist}~\cite{Pugh90} of $n$ nodes is organized as a series
of $O(\log n$) sorted linked lists, each specifying shortcuts of certain length.  
The bottom-level list contains all the nodes, each of the higher-level lists
contains a sublist of the lower-level list.  
A  traversal starts with the top-level list having the longest
``hops'' and goes to lower lists with smaller hops 
as the node with smallest key $k'\geq k$  get closer. 

A \emph{binary search tree} represents data items in the form of a binary
tree. Every node in the tree stores a key-value pair, and the left
descendant of a non-leaf node with key $k$ roots a subtree storing
all nodes with keys less than $k$, while the right
descendant roots a subtree storing all nodes with keys greater than $k$.
%\textbf{SR: No solution yet for BSTs}
%
Note that, for simplicity, we do not consider \emph{rebalancing} operations used by
%[[VG: skiplists are not rebalanced
%skiplists and
%]] 
balanced trees for maintaining the desired bounds on the
traverse complexity. Though crucial in practice, the rebalancing operations are not
important for our comparative analysis of concurrency properties of
synchronization techniques.

\paragraph{Non-serializable concurrency}
There is a straightforward LSL implementation of any data structure in
$\D$ in which updates (\textit{inserts} and \textit{deletes}) acquire
a lock on the root node and are thus
sequential. Moreover, they take exclusive locks on the set of nodes
they are about to modify ($k$-relevant sets for operations with
parameter $k$).   

%Additionally, update operations lock  
A \textit{find} operation uses \emph{hand-over-hand}
\emph{shared} locking~\cite{Wei88}: at each moment of time,
the operation holds shared locks on all outgoing edges for the
currently visited node $a$. To visit a new node $b$ (recall that $b$ must
be a descendant of $a$), it acquires shared locks
on the new node's descendants and then releases the shared lock on $a$.   
Note that just before a  $\textit{find}(k)$ operation returns the
result, it holds shared locks on the $k$-relevant set.   

This way updates always take place sequentially, in the order of their
acquisitions of the root lock.
A $\textit{find}(k)$ operation is linearized at any
point of its execution when it holds shared locks on the $k$-relevant set. 
Concurrent operations that do not contend on the same locks can be
arbitrarily ordered in a linearization.

The resulting \emph{HOH-find} implementation is described in
Algorithm~\ref{alg:hohfind}. 
The fact that the operations acquire (starvation-free) locks in the
order they traverse the directed acyclic graph implies that:   
%updates protect the data they modify using (starvation-free) exclusive locks
%and \textit{find} operations use shared locks implies that:
%
\begin{algorithm*}[t]
\caption{Abstract \emph{HOH-find} implementation of a
  {\dd} structure defined by
$(\tau_{op}$, $\mu_{insert(k,v)}$, $\nu_{delete(k)})$,
  $op\in\{\textit{insert}(k,v),\textit{delete}(k),\textit{find}(k)\}$,
  $k\in\Nat$, $v\in V$.}
\label{alg:hohfind}
  \begin{algorithmic}[1]
  	\begin{multicols}{2}
  	{\scriptsize
	
	\Part{Shared variables}{
		\State $\mathcal{G}$, initially $\ms{root}$ \Comment{Shared DAG}
	}\EndPart
	
	\Statex

	\Part{$\lit{find}(k)$}{

	         \State $\ms{G} \gets \emptyset$; $a \gets \{\ms{root}\}$ 
                 \State $a.\ms{lock-shared}()$
%                 \State $a  \gets \tau_{\textit{find}(k)}(a,G)$
                 \While{$a\neq \emptyset$} 
               %\Comment{If the current value is less than $v$}
                       \State $\forall (a,b)\in\mathcal{G}$: $b.\ms{lock-shared}()$
                       \State $\forall (a,b)\in\mathcal{G}$: $G \gets G\cup (a,b)$
                       \Comment{Explore new edges}
                       \State $\ms{last} \gets a$ 
		       \State $a  \gets \tau_{\textit{find}(k)}(a,G)$
                       \State $\forall (\ms{last},b)\in\mathcal{G}$,
                       $b\neq a$: $b.\ms{unlock-shared}()$
                       \State $\ms{last}.\ms{unlock-shared}()$
	   	\EndWhile
		\If{$G$ contains a node with key $k$}
                       \Return $\true$
                        \EndReturn
                \Else  
                       \Return $\false$
                        \EndReturn
                \EndIf
   	}\EndPart

        \Statex

        \Part{$\lit{insert}(k,v)$}{
               \State $\ms{root}.\ms{lock}()$ 
     	         \State $\ms{G} \gets \emptyset$; $a \gets \{\ms{root}\}$ 
           \While{$a\neq \emptyset$} 
               %\Comment{If the current value is less than $v$}
                      \State $\forall (a,b)\in\mathcal{G}$: $G \gets G\cup (a,b)$
                       \Comment{Explore new edges}
                       \State $\ms{last} \gets a$ 
		       \State $a  \gets \tau_{\textit{insert}(k,v)}(a,G)$
                       \State $\forall (\ms{last},b)\in\mathcal{G}$,
  	   	\EndWhile 

		 \If{$G$ contains no node with key $k$}
                	\State $a \gets\textit{create-node}(k,v)$
                        \State $\forall b$ such that
                        $\exists (b,c)\in\mu_{k}(G,a)$,  $b.\ms{lock}()$  
                        \State $\mathcal{G} \gets
                         \mathcal{G}\cup\mu_{(k,v)}(G,a)$
                        \Comment{Link $a$ to $\mathcal{G}$}
                       \State $\forall b$ such that
                        $\exists (b,c)\in\mu_{k}(G,a)$,  $b.\ms{unlock}()$  
                         \State $\ms{root}.\ms{unlock}()$
                        \Return $\true$
                        \EndReturn
                \Else  
                        \State  $\ms{root}.\ms{unlock}()$      
                        \Return $\false$
                        \EndReturn
                \EndIf
   	}\EndPart

	\newpage
	
     \Statex
     \Statex
     \Statex

      \Part{$\lit{delete}(k)$}{
               \State $\ms{root}.\ms{lock}()$  
     	         \State $\ms{G} \gets \emptyset$; $a \gets \{\ms{root}\}$ 
         \While{$a\neq \emptyset$} 
               %\Comment{If the current value is less than $v$}
                       \State $\forall (a,b)\in\mathcal{G}$: $G \gets G\cup (a,b)$
                       \Comment{Explore new edges}
                       \State $\ms{last} \gets a$ 
		       \State $a  \gets \tau_{\textit{delete}(k)}(a,G)$
                       \State $\forall (\ms{last},b)\in\mathcal{G}$,
	   	\EndWhile 

		\If{$G$ contains node $a$ with key $k$}
                	\State  remove $a$ and all edges to/from $a$ from $\mathcal{G}$ 	
                       \State $\forall b$  such that
                         $\exists (b,c)\in\nu_{k}(G,a)$,  $b.\ms{lock}()$  
                        \State $\mathcal{G}
                        \gets\mathcal{G}\cup\nu_{k}(G,a)$ \Comment{Shortcut edges}
                        \State $\forall b$ such that
                        $\exists (b,c)\in\nu_{k}(G,a)$,  $b.\ms{unlock}()$  
                         \State $\ms{root}.\ms{unlock}()$
                        \Return $\true$
                        \EndReturn
                \Else  
                        \State  $\ms{root}.\ms{unlock}()$      
                        \Return $\false$
                        \EndReturn
                \EndIf
   	}\EndPart
	
	}
	\end{multicols}
  \end{algorithmic}
\end{algorithm*}
\begin{theorem}\label{th:nonser}
The \emph{HOH-find} algorithm is a starvation-free LSL implementation of a
{\dd} structure. 
\end{theorem}
\begin{proof}
%\textbf{PK: update}
Take any execution $E$ of the algorithm. 
%Since the executions of \textit{find} operations is   
The subsequence of $E$ consisting of the events of update operations
is serializable (and, thus, locally serializable).
Since a $\textit{find}(k)$ operation protects its visited node and all its
outgoing edges with a shared lock and a concurrent update with a key protect
their $k$-relevant sets with an exclusive lock, $\textit{find}(k)$
observes the effects of updates as though they took place in a
sequential execution---thus local serializability. 

Let $H$ be the history of $E$. To construct a linearization of $H$, we
start with a sequential history $S$ that
orders all update operations in $H$ in the order in which they acquire
locks on the root. By construction, $S$ is legal.
A $\textit{find}(k)$ operation that returns $\true$ can only reach a node if
a node with key $k$ 
was reachable from the root at some point
during its interval. 
Similarly, if $\textit{find}(k)$ operation returns $\false$, then it would only fail to reach a node if it
was made unreachable from the root at some point during its interval.
Thus, every successful (resp., unsuccessful) $\textit{find}(k)$ operation $op$ can be inserted in $S$ after the
latest update operation that does not succeed in the real-time order
in $E$ and after which a node $k$ is reachable (resp., unreachable). 
By construction, the resulting sequential history is legal. 
\end{proof}
As we show in Section~\ref{sec:pessimistic}, the implementation is
however not (safe-strict) serializable.
%
%
%

%
%!TEX root = 0main.tex
\section{Pessimism vs. serializable optimism}
\label{sec:pessimistic}
In this section, we show that, with respect to
\dd{} structures, pessimistic locking and optimistic synchronization 
providing safe-strict serializability are \emph{incomparable}, once we focus on
LS-linearizable implementations. 

%[[PK16 rephrased and moved
\ignore{
More precisely, for any \dd{} structure, 
\begin{enumerate}
\item[\RNum{1}]There exists a schedule that is
rejected by \emph{any} pessimistic implementation, but accepted by
certain optimistic strictly serializable ones \label{res1}, and 
\item[\RNum{2}]There exists a schedule that is
rejected by \emph{any} serializable implementation but accepted by
a certain pessimistic one (\emph{HOH-find}, to be concrete) \label{res2}.  
\end{enumerate}
}
%]]

%[[PK16 moved from Sec 3
\subsection{Classes $\mathcal{P}$ and $\mathcal{SM}$}

A \emph{synchronization technique} is a set of concurrent implementations.
% [[PK: too long
%We now define two synchronization techniques: one encompasses a subclass of
%optimistic implementations while the other characterizes pessimistic implementations.
We define below a specific optimistic synchronization technique and then
a specific pessimistic one.
%
%\paragraph{The class $\mathcal{P}$.}

\paragraph{$\mathcal{SM}$: serializable optimistic}
%\paragraph{The class $\mathcal{SM}$.}
Let $\alpha$ denote the execution of a concurrent implementation and
$\ms{ops}(\alpha)$, 
the set of operations each of which performs at least one event in $\alpha$.
%Let $\ms{cseq}(\alpha)$ be the subsequence of a finite execution $\alpha$ derived by removing all 
%events from every transaction that is either incomplete or aborted in $\alpha$.
Let ${\alpha}^k$ denote the prefix of $\alpha$ up to the last event of operation $\pi_k$.
Let $\ms{Cseq}(\alpha)$ denote the set of subsequences of ${\alpha}$  that
consist of all the events of operations that are complete in $\alpha$. 
We say that $\alpha$ is \emph{strictly serializable} if 
there exists a legal sequential execution $\alpha'$ equivalent to
a sequence in $\sigma\in\ms{Cseq}(\alpha)$
such that $\rightarrow_{\sigma} \subseteq \rightarrow_{\alpha'}$. 

This paper focuses on optimistic implementations that are strictly
serializable and, in addition, guarantee that every operation (even aborted or incomplete) observes correct (serial)
behavior.  
More precisely, an execution $\alpha$ is  \emph{safe-strict serializable} if
(1) $\alpha$ is strictly serializable, and
(2) for each operation $\pi_k$, % that is incomplete or returned
%$\bot$ in $\alpha$, 
there exists a legal sequential execution
$\alpha'=\pi_0\cdots \pi_i\cdot \pi_k$ and
$\sigma\in\ms{Cseq}(\alpha^k)$ such that $\{\pi_0,\cdots, \pi_i\}
\subseteq \ms{ops}(\sigma)$ and $\forall \pi_m\in \ms{ops}(\alpha'):{\alpha'}|m={\alpha^k}|m$.

Safe-strict serializability captures nicely both local serializability
and linearizability. 
If we transform a sequential implementation
$\id{IS}$ of a type $\tau$ into a \emph{safe-strict serializable} concurrent one, 
we obtain an LSL implementation of $(\id{IS},\tau)$. 
Thus, the following lemma is immediate.
 \begin{lemma}
 Let $I$ be a safe-strict serializable implementation of $(\id{IS},\tau)$.
 Then, $I$ is \LS-linearizable with respect to $(\id{IS},\tau)$.
 \end{lemma}
Indeed, 
% by running each operation of $\id{IS}$ within a transaction of
% a safe-strict serializable TM, 
we make sure that completed operations witness the same execution of $\id{IS}$, and
every operation that returned $\bot$ is consistent with some
execution of $\id{IS}$ based on previously completed operations. 
Formally, $\mathcal{SM}$ denotes the set of optimistic, safe-strict serializable LSL
implementations.

\paragraph{$\mathcal{P}$: deadlock-free pessimistic}
%This denotes the set of \emph{deadlock-free} pessimistic
%LSL implementations: 
Assuming that no process stops taking steps of
its algorithm in the middle of a high-level operation,
at least one of the concurrent operations return a matching response~\cite{HS11-progress}.
Note that $\mathcal{P}$ includes implementations that are not necessarily safe-strict serializable. 

%%%%%%%%%%%%%%%%%%%%%%%%%%%%%%%%%%%%%%%%%%%%%%%
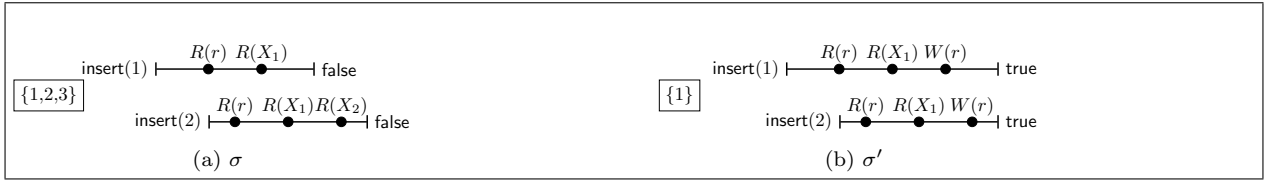
\begin{figure*}[t]
 \subfloat[$\sigma$\label{sfig:inv-1}]{\scalebox{0.7}[0.7]{\begin{tikzpicture}
\node (r1) at (1,0) [] {};
\node (r2) at (2,0) [] {};

\node (r3) at (1.5,-1) [] {};
\node (r4) at (2.5,-1) [] {};
\node (r5) at (3.5,-1) [] {};

\draw (r1) node [above] { {$R(r)$}};
\draw (r2) node [above] { {$R(X_1)$}};

\draw (r3) node [above] { {$R(r)$}};
\draw (r4) node [above] { {$R(X_1)$}};
\draw (r5) node [above] { {$R(X_2)$}};

\begin{scope}   
\draw [|-|,thick] (0,0) node[left] {$\lit{insert}(1)$} to (3,0) node[right] {$\lit{false}$};
\draw  [fill=black, radius=0.1]  (1,0) circle (.6ex);
\draw  [fill=black, radius=0.1]  (2,0) circle (.6ex);

\draw [|-|,thick] (1,-1) node[left] {$\lit{insert}(2)$} to (4,-1) node[right] {$\lit{false}$};
\draw  [fill=black, radius=0.1]  (1.5,-1) circle (.6ex);
\draw  [fill=black, radius=0.1]  (2.5,-1) circle (.6ex);
\draw  [fill=black, radius=0.1]  (3.5,-1) circle (.6ex);

\end{scope}
\node[draw,align=right] at (-2,-.5) { \{1,2,3\}};
\end{tikzpicture}}}\hspace{30mm}
 \subfloat[$\sigma'$\label{sfig:inv-2}]{\scalebox{0.7}[0.7]{\begin{tikzpicture}
\node (r1) at (1,0) [] {};
\node (r2) at (2,0) [] {};
\node (w1) at (3,0) [] {};

\node (r3) at (1.5,-1) [] {};
\node (r4) at (2.5,-1) [] {};
\node (r5) at (3.5,-1) [] {};

\draw (r1) node [above] { {$R(r)$}};
\draw (r2) node [above] { {$R(X_1)$}};
\draw (w1) node [above] { {$W(r)$}};

\draw (r3) node [above] { {$R(r)$}};
\draw (r4) node [above] { {$R(X_1)$}};
\draw (r5) node [above] { {$W(r)$}};

\begin{scope}   
\draw [|-|,thick] (0,0) node[left] {$\lit{insert}(1)$} to (4,0) node[right] {$\lit{true}$};
\draw  [fill=black, radius=0.1]  (1,0) circle (.6ex);
\draw  [fill=black, radius=0.1]  (2,0) circle (.6ex);
\draw  [fill=black, radius=0.1]  (3,0) circle (.6ex);

\draw [|-|,thick] (1,-1) node[left] {$\lit{insert}(2)$} to (4,-1) node[right] {$\lit{true}$};
\draw  [fill=black, radius=0.1]  (1.5,-1) circle (.6ex);
\draw  [fill=black, radius=0.1]  (2.5,-1) circle (.6ex);
\draw  [fill=black, radius=0.1]  (3.5,-1) circle (.6ex);

\end{scope}
\node[draw,align=right] at (-2,-.5) { \{1\}};
\end{tikzpicture}}}
 \caption{\small{%
(a) a history of integer set (implemented as linked list or binary search tree) exporting schedule $\sigma$, with initial state
   $\{1,2,3\}$ ($r$ denotes the root node); 
%[[PK15
%but not by any $I\in \mathcal{P}$; 
(b) a history exporting a problematic schedule $\sigma'$, with initial state 
   $\{3\}$, which should be accepted by any $I\in\mathcal{P}$ if it accepts $\sigma$}}\label{fig:ex2}%
\vspace{-0.35mm}
\end{figure*}
%%%%%%%%%%%%%%%%%%%%%%%%%%%%%%%%%%%%%%%%%%%%%%%
\subsection{Suboptimality of pessimistic implementations}

 We show now that for any \dd{} structure, there exists a schedule that is
rejected by \emph{any} pessimistic implementation, but accepted by
certain optimistic strictly serializable ones.
To prove this claim, we derive a safe-strict serializable schedule that cannot be accepted by any
implementation in $\P$ using the \emph{non-triviality} property of
{\dd} structures. It turns out that we can schedule the traverse
phases of two $\textit{insert}(k)$ operations in parallel until they
are about to check if a node with key $k$ is in the set or not. 
If it is, both operations may safely return $\false$ (schedule $\sigma$). However, if the node
is not in the set, in a pessimistic implementation, both operations
would have to modify outgoing edges of the same node $a$ and, if we
want to provide local serializability, both return $\true$, violating
linearizability (schedule $\sigma'$). 

In contrast, an optimistic implementation may simply abort one of the
two operations in case of such a conflict, by accepting the (correct) schedule $\sigma$
and rejecting the (incorrect) schedule $\sigma'$.  

\textbf{Proof intuition.}
We first provide an intuition of our results in the context of the \emph{integer set} implemented as a 
\emph{sorted linked list} or \emph{binary search tree}.
The set type is a special case of the dictionary which stores a set of integer values,
initially empty, and exports
operations $\textit{insert}(v)$, $\textit{remove}(v)$, $\textit{find}(v)$; $v \in \mathbb{Z}$.
The update operations, $\textit{insert}(v)$ and $\textit{remove}(v)$, return
a boolean response, $\true$ if and only if $v$ is absent (for
$\textit{insert}(v)$) or present (for $\textit{remove}(v)$) in the set.  
After $\textit{insert}(v)$ is complete, $v$ is present in the set, and 
after $\textit{remove}(v)$ is complete, $v$ is absent in the set.
The $\textit{find}(v)\}$ returns a boolean a boolean, $\true$ if and
only if $v$ is present in the set.

An example of schedules $\sigma$ and $\sigma'$ of the set is
given in Figure~\ref{fig:ex2}.
%Figure~\ref{fig:ex2} depicts the proof of Theorem~\ref{th:mpl} in the context of the list-based set.
We show that the schedule $\sigma$ depicted
in Figure~\ref{fig:ex2}(a) is not accepted by any implementation in $\mathcal{P}$.
Suppose the contrary and let $\sigma$ be exported by an execution $\alpha$. 
Here $\alpha$ starts with three sequential $\lit{insert}$ operations with
parameters $1$, $2$, and $3$. The resulting ``state'' of the set is
$\{1,2,3\}$, where value $i\in \{1,2,3\}$ is stored in node $X_i$.   
Suppose, by contradiction, that some $I\in \mathcal{P}$ accepts $\sigma$. 
We show that $I$ then accepts the schedule $\sigma'$ depicted in Figure~\ref{fig:ex2}(b), which starts with a sequential
execution of $\lit{insert}(3)$ storing value $3$ in node $X_1$. 
We can further extend $\sigma'$ with a complete $\lit{find}(1)$ (by deadlock-freedom of $\mathcal{P}$)
that will return $\false$ (the node inserted to the list by $\lit{insert}(1)$
is lost)---a contradiction since $I$ is linearizable with respect to \emph{set}. 
%of Appendix~\ref{app:seq}.       

The formal proof follows.
\begin{theorem}
\label{th:mpl}
%$\ms{M} \not\preceq_{(\LL, \ms{set})} \ms{PL}$.
Any abstraction in $\D$ has a strictly serializable schedule 
that is not accepted by \emph{any} implementation in
$\mathcal{P}$, but accepted by an implementation in $\mathcal{SM}$.
\end{theorem}
\begin{proof}
Consider the key $k$ and the states $G$ and $G'$ satisfying the conditions of
the \textit{non-triviality} property.

Consider the schedule that begins with a serial sequence of operations
bringing the {\dd} to state $G$   (executed by a single process). 
Then schedule the traverse phases of two identical $\textit{insert}(k,v)$ operations  executed by
\emph{new} (not yet participating) processes $p_1$ and $p_2$ concurrently so that they
perform identical steps (assuming that, if they take randomized steps,
their coin tosses return the same values).   
Such an execution $E$ exists, since the traverse phases are
read-only.
But if we allow both insert operations to proceed (by deadlock-freedom), we obtain an execution
that is not LS-linearizable: both operations update the data
structure which can only happen in a successful insert. 
But, by the sequential specification of $D$, since node with
key $k$ belongs to $G$, at least one of
the two inserts must fail.   
Therefore, a pessimistic implementation, since it is not allowed to
abort an operation, cannot accept the corresponding
schedule $\sigma$.

Now consider the serial sequence of operations bringing $D$
to state $G'$ (executed by a single process) and extend it with traverse
phases of two concurrent $\textit{insert}(k,v)$ operations executed by
new processes $p_1$ and $p_2$.  
The two traverse phases produce an execution $E'$  which is
indistinguishable to $p_1$ and $p_2$ from $E$ up to their last read
operations. Thus, if a pessimistic implementation accepts the corresponding
schedule $\sigma$, it must also accept $\sigma'$, violating LS-linearizability. 

Note, however, that an extension of $E'$ in which both inserts
complete by returning $\false$ is LS-linearizable. 
Moreover, any \emph{progressive} (e.g., using progressive opaque
transactional memory) optimistic strictly serializable implementation using 
will accept $\sigma'$.    
\end{proof}
%
%%%%%%%%%%%%%%%%%%%%%%%%%%%%%%%%%%%%%%%%%%%%%%%%%%%%%%%%%%%%%%%
\begin{figure*}[t]
\scalebox{0.7}[0.7]{\begin{tikzpicture}
\node (r1) at (1,0) [] {};
\node (r2) at (2,0) [] {};
\node (r3) at (3,0) [] {};
\node (r4) at (13,0) [] {};
\node (r5) at (14,0) [] {};

\node (i1) at (5,-1) [] {};
\node (i2) at (6,-1) [] {};

\node (i3) at (9,-2) [] {};
\node (i4) at (10,-2) [] {};
\node (w1) at (11,-2) [] {};

\draw (r1) node [above] {{$R(r)$}};
\draw (r2) node [above] {{$R(X_1)$}};
\draw (r3) node [above] { {$R(X_3)$}};
\draw (r4) node [above] { {$R(X_4)$}};
\draw (r5) node [above] { {$R(X_5)$}};

\draw (i1) node [above] {{$R(r)$}};
\draw (i2) node [above] {{$W(X_1)$}};

\draw (i3) node [above] {{$R(r)$}};
\draw (i4) node [above] {{$R(X_1)$}};
\draw (w1) node [above] {{$W(X_4)$}};

\begin{scope}   
\draw [|-|,thick] (0,0) node[left] {$\lit{find}(5)$} to (15,0) node[right] {$\lit{true}$};
\draw  [black,fill=black, radius=0.1] (1,0) circle (.6ex);
\draw  [fill=black, radius=0.1]  (2,0) circle (.6ex);
\draw   [fill=black, radius=0.1]  (3,0) circle (.6ex);
\draw   [fill=black, radius=0.1]  (13,0) circle (.6ex);
\draw   [fill=black, radius=0.1]  (14,0) circle (.6ex);

\end{scope}

\begin{scope}   
\draw [|-|,thick] (4,-1) node[left] {$\lit{insert}(2)$} to (7,-1) node[right] {$\lit{true}$};
\draw   [fill=black, radius=0.1]  (5,-1) circle (.6ex);
\draw   [fill=black, radius=0.1]  (6,-1) circle (.6ex);

\draw [|-|,thick] (8,-2) node[left] {$\lit{insert}(5)$} to (12,-2) node[right] {$\lit{true}$};
\draw   [fill=black, radius=0.1]  (9,-2) circle (.6ex);
\draw   [fill=black, radius=0.1]  (10,-2) circle (.6ex);
\draw   [fill=black, radius=0.1]  (11,-2) circle (.6ex);

\end{scope}

\end{tikzpicture}}
\caption{%\small
 {A concurrency scenario for a set, initially $\{1,3,4\}$, where value $i$ is stored at node $X_i$:
   $\lit{insert}(2)$ and $\lit{insert}(5)$ can proceed
   concurrently with $\lit{find}(5)$. The history is
   LS-linearizable but not serializable; yet accepted by HOH-find. (Not all read-write on nodes is presented
   here.)}}\label{fig:ex1}%
\end{figure*}
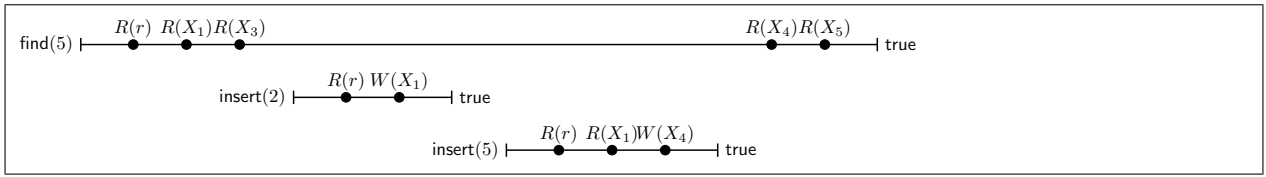
%%%%%%%%%%%%%%%%%%%%%%%%%%%%%%%%%%%%%%%%%%%%%%%%%%%%%%%%%%%%%
%
\subsection{Suboptimality of serializable optimism}
\label{sec:seropt}
We show below that for any \dd{} structure, there exists a schedule that is
rejected by \emph{any} serializable implementation but accepted by
a certain pessimistic one (\emph{HOH-find}, to be concrete).

\textbf{Proof intuition.}
We first illustrate the proof in the context of the integer set.
Consider a schedule $\sigma_0$ of a concurrent set implementation depicted in 
Figure~\ref{fig:ex1}.
We assume here that the set initial state is $\{1,3,4\}$.
Operation $\lit{find}(5)$ is concurrent, first with 
operation $\lit{insert}(2)$ and then with operation $\textsf{insert}(5)$. 
The history is not serializable:
$\lit{insert}(5)$ sees 
the effect of $\lit{insert}(2)$ because $R(X_1)$ by $\lit{insert}(5)$ returns the value
of $X_1$ that is updated by $\lit{insert}(2)$ and
thus should be serialized after it. But $\lit{find}(5)$ misses
node with value $2$ in the set, but must read the value of $X_4$ that is updated by
$\lit{insert}(5)$ to perform the read of $X_5$, \emph{i.e.}, the node created by $\lit{insert}(5)$. 
Thus, $\sigma_0$ is not (safe-strict) serializable.
However, this history is LSL since each of the three local histories is consistent with some
sequential history of the integer set.
However, there exists an execution of our HOH-find implementation that exports $\sigma_0$
since there is no read-write conflict on any two consecutive nodes accessed.

To extend the above idea to any search structure, we use the
\emph{non-triviality} property of data structures in $\D$.
There exist a state $G'$ in which there is exactly one edge $(a,b)$ in $G'$ such that
$b$ has key $k$.
We schedule a $op_f=\textit{find}(k)$ operation concurrently with two
consecutive delete operations: the first one, $op_{d1}$, deletes one of the nodes
explored by $op_f$ before it reaches $a$ (such a node exists by the
\textit{non-triviality} property), and the second one, $op_{d2}$
deletes the node with key $k$ in $G'$.
We make sure that  $op_f$ is not affected by $op_{d1}$ (observes an update to some node $c$
in the graph) but is affected 
by $op_{d2}$ (does not observe $b$ in the graph). The resulting
schedule is not strictly serialializable (though linearizable). 
But our HOH-find implementation in $\P$ will accept it.   
\begin{theorem}
\label{th:plm}
%$\ms{PL} \not\preceq_{(\LL, \ms{set})} \ms{M}$.
For any abstraction in $D\in\D$, there exists an implementation in
$\mathcal{P}$ that accepts a non-strictly serializable schedule.
\end{theorem}
\begin{proof}
Consider the \emph{HOH-find} implementation described in
Algorithm~\ref{alg:hohfind}. Take the key $k$ and the states $G$ and $G'$ 
satisfying the conditions of the \emph{non-triviality} property.

Now we consider the following execution. 
Let $op_f=\textit{find}(k)$ be applied to an execution resulting in
$G'$ (that contains a node with key $k$)  and run $op_f$ 
until it reads the only node $a=(k',v')$ in $G$ that points to a node $b=(k,v)$ in state $G'$. 
Note that since $R_k(G)=R_k(G')$, the
operation cannot distinguish the execution from than one starting with $G$.
 
The \emph{non-triviality} property requires that the shortest path
from the root to $k$ to $a$ in $R_k(G)$ is of length at least
two. Thus, the set of nodes explored by $op_f$ passed through at least
one node $c=(k'',v'')$ in addition to $a$. Now we schedule two 
complete delete operations executed by another process: first 
$\textit{del}_c= \textit{delete}(k'')$ which removes $c$,
followed by $\textit{del}_b=\textit{delete}(k)$ which removes $b$.  
Now we wake up $op_f$ and let it read $a$, find out that no node with
key $k$ is reachable, and return $\false$

Suppose, by contradiction, that the resulting execution is
strictly serializable. Since $op_f$ has witnessed the presence of some node $c$
on the path from the root to $a$ in the DAG,
$op_f$ must precede $\textit{del}_c$ in any serializaton.
Now $\textit{del}_b$ affected the response of $op_f$, it must precede $op_f$ in any
serialization. Finally, $\textit{del}_c$ precedes $\textit{del}_b$ in
the real-time order and, thus must precede $\textit{del}_b$ in any
serialization. The resulting cycle implies a contradiction.
\end{proof}%
Since any strictly serializable optimistic implementation only produces
strictly serializable executions, from Theorem~\ref{th:plm} we deduce that there
is a schedule accepted by a pessimistic algorithm that no strictly serializable
optimistic one can accept. 
Therefore, Theorems~\ref{th:mpl} and~\ref{th:plm} imply that, when applied to
\dd{} structures and in terms of concurrency, the strictly  serializable
optimistic approach is incomparable with pessimistic locking.  
%[[PK
As a corollary, none of these two techniques can be
concurrency-optimal.

%
%========================================================
%========================================================
%
%!TEX root = 0main.tex
\section{Related work}
\label{sec:related}
Sets of accepted schedules are commonly used as a
metric of concurrency provided by a shared memory
implementation.
For static database transactions, 
Kung and Papadimitriou~\cite{KP79} use the metric to 
capture the parallelism of a locking scheme,
%by the amount of 
%interleavings of operations it could accept without having to 
%reschedule them~\cite{KP79}. 
While acknowledging that the metric is theoretical, they 
insist that it may
have ``practical significance as
well, if the schedulers in question have relatively small
scheduling times as compared with waiting and execution
times.'' 
Herlihy~\cite{Her90} employed the metric to compare various
optimistic and pessimistic synchronization techniques using
commutativity
%\emph{commutativity}~\cite{Wei88} 
of operations constituting high-level transactions.   
A synchronization technique is implicitly considered in~\cite{Her90} as highly
concurrent, namely ``optimal'',
if no other technique accepts more schedules. 
%Unlike~\cite{KP79,Her90}, 
By contrast, we focus here on a \emph{dynamic} model where the scheduler cannot 
use the prior knowledge of all the shared addresses to be accessed. 
% exploit as much syntactic information 
%as in a static model where all the shared addresses to be accessed are explicitly known before execution.
Also, unlike~\cite{KP79,Her90}, %in the transactional memory context, 
we require \emph{all} operations, including aborted ones, to observe (locally) consistent states.
%[[ V: two "which prepositions"
%[[PK do we really this here?
%As we confirm experimentally, 
%our (provably optimal) optimistic implementation incurs negligible scheduling overhead, which makes the motivation of the metric proposed
%in~\cite{KP79} applicable. %, which is also supported by our performance evaluation. 
%]]

Gramoli \emph{et al}.~\cite{GHF10} defined a concurrency metric, the \emph{input
acceptance}, as the ratio of committed transactions over aborted
transactions when TM executes the given schedule.   
Unlike our metric, input acceptance does not apply to
lock-based programs. 

Optimal concurrency is related to the notion of
\emph{permissiveness}~\cite{GHS08-permissiveness}, 
originally defined for transactional memory. In \cite{GHS08-permissiveness},
a TM is defined to be \emph{permissive with respect to serializability} 
if it aborts a transaction only if committing it
would violate serializability.
Thus, an operation of a concurrent data structure may be aborted only if the
execution of the encapsulated transaction is not serializable.
In contrast, our framework for analyzing concurrency
is independent of the synchronization technique and a concurrency-optimal implementation
accepts all correct interleavings of reads and writes of sequential operations (not just
high-level responses of the operations).

David et al.~\cite{DGT15} devised a pragmatic methodology for evaluating performance
of a concurrent data structure via a comparison with the performance
of its sequential counterpart. The closer is the throughput of a
concurrent algorithm is to that of its (inconsistent) sequential variant, the more
``concurrent'' the algorithm is.   
In contrast, the formalism proposed in this paper allows for relating concurrency
properties of various concurrent algorithms. 
   
Our definition of \dd{} data structures is based on the paper by Chaudhri and Hadzilacos~\cite{CH98} who studied them in the context
of dynamic databases.

Safe-strict serializable implementations ($\mathcal{SM}$) require that every transaction (even aborted and incomplete) observes
``correct'' serial behavior. 
It is weaker than popular TM correctness conditions like opacity~\cite{tm-book} and its
relaxations like \emph{TMS1}~\cite{TMS09} and \emph{VWC}~\cite{damien-vw},  
Unlike TMS1, we do not require the \emph{local} serial executions
%observed by a transaction that may abort 
to always respect the real-time order among transactions. % that have committed or trying to commit.
Unlike VWC, we model transactional
operations as intervals with an invocation and a response and does not
assume unique writes (needed to define causal past in VWC). 
Though weak, $\mathcal{SM}$ still allows us 
to show that the resulting optimistic LSL implementations 
reject some schedules accepted by pessimistic locks. 
% However, we can easily extend our results to show that even opaque optimistic implementations accept
% some schedules rejected by any pessimistic algorithm.%
%
%
%
%

%
%!TEX root = 0main.tex
\section{Concluding remarks}
\label{sec:conc}

In this paper, we presented a formalism for reasoning about the relative power of optimistic and pessimistic
synchronization techniques in exploiting concurrency in \dd{}
structures. 
We expect our formalism to have practical impact as the \dd{} structures are among the
most commonly used concurrent data structures, including trees, linked lists, skip lists that 
implement various abstractions ranging from key-value stores to sets and multi-sets. 

Our results on the relative concurrency of $\mathcal{P}$
and $\mathcal{SM}$ imply that none of these synchronization techniques
might enable an optimally-concurrent algorithm.
Of course, we do not claim that our concurrency metric necessarily captures 
efficiency, as it does not account for other factors, 
like cache sizes, cache coherence protocols, or computational costs of 
validating a schedule, which may also affect performance on
multi-core architectures.
%Our results only relate to concurrency and investigating whether this concurrency translates into performance is an interesting research question.
%
%
%In a concurrent paper~\cite{CONCURP-TR}, however, 
In~\cite{GKR15} we already described a
\emph{concurrency-optimal} implementation of the linked-list set
abstraction that combines  the advantages of $\mathcal{P}$, 
namely the semantics awareness, with the advantages of $\mathcal{SM}$, namely the ability to 
restart operations in case of conflicts.
%The implementation is not only LS-linearizable, but also 
%accepts all \emph{correct} (LS-linearizable) schedules. 
We recently observed empirically %experimentally shown 
that this optimality can result in higher performance than 
state-of-the-art algorithms~\cite{HHL+05,harris-set,michael-set}.     
%This paper, considers the notion of relative concurrency 
%in the general context of search data structures, which allows us to
%compare the concurrency properties of various synchronization
%techniques that can be used to implement concurrent dictionaries.        
%
Therefore, our findings motivate the search for concurrency-optimal
algorithms. This study not only improves our understanding
of designing concurrent data structures, but might lead to more efficient implementations.

\bibliography{references}

\def\noopsort#1{} \def\No{\kern-.25em\lower.2ex\hbox{\char'27}}
  \def\no#1{\relax} \def\http#1{{\\{\small\tt
  http://www-litp.ibp.fr:80/{$\sim$}#1}}}
\begin{thebibliography}{10}

\bibitem{PLE12}
Y.~Afek, A.~Matveev, and N.~Shavit.
\newblock Pessimistic software lock-elision.
\newblock In {\em Proceedings of the 26th International Conference on
  Distributed Computing}, DISC'12, pages 297--311, Berlin, Heidelberg, 2012.
  Springer-Verlag.

\bibitem{AFHHT07}
M.~K. Aguilera, S.~Fr{\o}lund, V.~Hadzilacos, S.~L. Horn, and S.~Toueg.
\newblock Abortable and query-abortable objects and their efficient
  implementation.
\newblock In {\em PODC}, pages 23--32, 2007.

\bibitem{AW04}
H.~Attiya and J.~Welch.
\newblock {\em Distributed Computing. Fundamentals, Simulations, and Advanced
  Topics.}
\newblock John Wiley \& Sons, 2004.

\bibitem{CH98}
V.~K. Chaudhri and V.~Hadzilacos.
\newblock Safe locking policies for dynamic databases.
\newblock {\em J. Comput. Syst. Sci.}, 57(3):260--271, 1998.

\bibitem{norec}
L.~Dalessandro, M.~F. Spear, and M.~L. Scott.
\newblock {NOrec}: streamlining {STM} by abolishing ownership records.
\newblock In {\em PPOPP}, pages 67--78, 2010.

\bibitem{DGT15}
T.~David, R.~Guerraoui, and V.~Trigonakis.
\newblock Asynchronized concurrency: The secret to scaling concurrent search
  data structures.
\newblock In {\em Proceedings of the Twentieth International Conference on
  Architectural Support for Programming Languages and Operating Systems,
  {ASPLOS} '15, Istanbul, Turkey, March 14-18, 2015}, pages 631--644, 2015.

\bibitem{TMS09}
S.~Doherty, L.~Groves, V.~Luchangco, and M.~Moir.
\newblock Towards formally specifying and verifying transactional memory.
\newblock {\em Electron. Notes Theor. Comput. Sci.}, 259:245--261, Dec. 2009.

\bibitem{FFR08}
P.~Felber, C.~Fetzer, and T.~Riegel.
\newblock Dynamic performance tuning of word-based software transactional
  memory.
\newblock In {\em PPoPP}, pages 237--246, 2008.

\bibitem{GG14}
V.~Gramoli and R.~Guerraoui.
\newblock Democratizing transactional programming.
\newblock {\em Commun. ACM}, 57(1):86--93, Jan 2014.

\bibitem{GHF10}
V.~Gramoli, D.~Harmanci, and P.~Felber.
\newblock On the input acceptance of transactional memory.
\newblock {\em Parallel Processing Letters}, 20(1):31--50, 2010.

\bibitem{GKR15}
V.~Gramoli, P.~Kuznetsov, S.~Ravi, and D.~Shang.
\newblock Brief announcement: A concurrency-optimal list-based set.
\newblock In {\em Distributed Computing - 29th International Symposium, {DISC}
  2015, Tokyo, Japan, October 7-9}, 2015.

\bibitem{GHS08-permissiveness}
R.~Guerraoui, T.~A. Henzinger, and V.~Singh.
\newblock Permissiveness in transactional memories.
\newblock In {\em DISC}, pages 305--319, 2008.

\bibitem{tm-book}
R.~Guerraoui and M.~Kapalka.
\newblock {\em Principles of Transactional Memory,Synthesis Lectures on
  Distributed Computing Theory}.
\newblock Morgan and Claypool, 2010.

\bibitem{harris-set}
T.~L. Harris.
\newblock A pragmatic implementation of non-blocking linked-lists.
\newblock In {\em DISC}, pages 300--314, 2001.

\bibitem{HHL+05}
S.~Heller, M.~Herlihy, V.~Luchangco, M.~Moir, W.~N. Scherer, and N.~Shavit.
\newblock A lazy concurrent list-based set algorithm.
\newblock In {\em OPODIS}, pages 3--16, 2006.

\bibitem{Her90}
M.~Herlihy.
\newblock Apologizing versus asking permission: optimistic concurrency control
  for abstract data types.
\newblock {\em ACM Trans. Database Syst.}, 15(1):96--124, 1990.

\bibitem{Her91}
M.~Herlihy.
\newblock Wait-free synchronization.
\newblock {\em ACM Trans. Prog. Lang. Syst.}, 13(1):123--149, 1991.

\bibitem{HS08-book}
M.~Herlihy and N.~Shavit.
\newblock {\em The art of multiprocessor programming}.
\newblock Morgan Kaufmann, 2008.

\bibitem{HS11-progress}
M.~Herlihy and N.~Shavit.
\newblock On the nature of progress.
\newblock In {\em OPODIS}, pages 313--328, 2011.

\bibitem{HW90}
M.~Herlihy and J.~M. Wing.
\newblock Linearizability: A correctness condition for concurrent objects.
\newblock {\em ACM Trans. Program. Lang. Syst.}, 12(3):463--492, 1990.

\bibitem{damien-vw}
D.~Imbs, J.~R.~G. de~Mend\'{\i}vil, and M.~Raynal.
\newblock Brief announcement: virtual world consistency: a new condition for
  stm systems.
\newblock In {\em PODC}, pages 280--281, 2009.

\bibitem{KP79}
H.~T. Kung and C.~H. Papadimitriou.
\newblock An optimality theory of concurrency control for databases.
\newblock In {\em SIGMOD}, pages 116--126, 1979.

\bibitem{michael-set}
M.~M. Michael.
\newblock High performance dynamic lock-free hash tables and list-based sets.
\newblock In {\em SPAA}, pages 73--82, 2002.

\bibitem{Pap79-serial}
C.~H. Papadimitriou.
\newblock The serializability of concurrent database updates.
\newblock {\em J. ACM}, 26:631--653, 1979.

\bibitem{Pugh90}
W.~Pugh.
\newblock Skip lists: {A} probabilistic alternative to balanced trees.
\newblock {\em Commun. {ACM}}, 33(6):668--676, 1990.

\bibitem{ST95}
N.~Shavit and D.~Touitou.
\newblock Software transactional memory.
\newblock In {\em PODC}, pages 204--213, 1995.

\bibitem{Wei88}
W.~E. Weihl.
\newblock Commutativity-based concurrency control for abstract data types.
\newblock {\em IEEE Trans. Comput.}, 37(12):1488--1505, 1988.

\bibitem{Wei86}
G.~Weikum.
\newblock A theoretical foundation of multi-level concurrency control.
\newblock In {\em PODS}, pages 31--43, 1986.

\bibitem{WV02-book}
G.~Weikum and G.~Vossen.
\newblock {\em Transactional Information Systems: Theory, Algorithms, and the
  Practice of Concurrency Control and Recovery}.
\newblock Morgan Kaufmann, 2002.

\bibitem{Yan84}
M.~Yannakakis.
\newblock Serializability by locking.
\newblock {\em J. ACM}, 31(2):227--244, 1984.

\end{thebibliography}

% \appendix
% \input{appendix}
\end{document}